\documentclass[10pt,conference]{IEEEtran}
\makeatletter
\def\ps@headings{%
\def\@oddhead{\mbox{}\scriptsize\rightmark \hfil \thepage}%
\def\@evenhead{\scriptsize\thepage \hfil \leftmark\mbox{}}%
\def\@oddfoot{}%
\def\@evenfoot{}}
\makeatother
\usepackage{lscape}
\usepackage{stfloats,amsmath,enumerate,setspace}
\usepackage[ruled,vlined]{algorithm2e}
\usepackage{amsthm}
\usepackage{fancybox}
\usepackage{amsmath}
\usepackage{algpseudocode}
\usepackage{amssymb}
\usepackage[demo]{graphicx}
\usepackage{caption}
\usepackage{subcaption}
\usepackage{flushend}

\newtheorem{theorem}{Theorem}

\newtheorem{lemma}{Lemma}

\newtheorem{observation}{Observation}
\usepackage{graphicx}
\usepackage{verbatim}

\theoremstyle{example}

\author
{
  \IEEEauthorblockN
  {
    Sukanta Bhattacharjee\IEEEauthorrefmark{1},  Ansuman Banerjee\IEEEauthorrefmark{1}, Tsung-Yi Ho\IEEEauthorrefmark{2}, Krishnendu Chakrabarty\IEEEauthorrefmark{3}, and Bhargab B. Bhattacharya\IEEEauthorrefmark{1}\\ 
  }
  \IEEEauthorblockA
  {\footnotesize
    \IEEEauthorrefmark{1}Nanotechnology Research Triangle, Indian Statistical Institute, Kolkata 700108, India, E-mail: 
    \{sukanta\_r, ansuman, bhargab\}@isical.ac.in
  }
  \IEEEauthorblockA
  {\footnotesize
    \IEEEauthorrefmark{2}Dept. CS \& Inform. Engg., National Cheng Kung University, Tainan, Taiwan  70101, E-mail: tyho@csie.ncku.edu.tw
  }
  \IEEEauthorblockA
  {\footnotesize
    \IEEEauthorrefmark{3}Dept. ECE, Duke University, Durham, NC 27708, USA, E-mail: krish@ee.duke.edu
  }
}

\title{\huge Algorithms for Producing Linear Dilution Gradient with Digital Microfluidics}

\date{}

\begin{document}

\maketitle

\begin{abstract}
Digital microfluidic (DMF) biochips  are now being extensively used to automate several biochemical laboratory protocols such as clinical analysis, point-of-care diagnostics, and polymerase chain reaction (PCR). In many biological assays, e.g., in bacterial susceptibility tests, samples and reagents are required in multiple concentration (or dilution) factors, satisfying certain ``gradient" patterns such as linear, exponential, or parabolic. Dilution gradients are usually prepared with continuous-flow microfluidic devices; however, they suffer from inflexibility, non-programmability, and from large requirement of costly stock solutions. DMF biochips, on the other hand, are shown to produce, more efficiently, a set of random dilution factors. However, all existing algorithms fail to optimize the cost or performance when a certain gradient pattern is required. In this work, we present an algorithm to generate any arbitrary linear gradient, on-chip, with minimum wastage, while satisfying a required accuracy in the concentration factor. We present new theoretical results on the number of \textit{mix-split} operations and \textit{waste} computation, and prove an upper bound on the storage requirement. The corresponding layout design of the biochip is also proposed. Simulation results on different linear gradients show a significant improvement in sample cost over three earlier algorithms used for the generation of multiple concentrations. 
 
\end{abstract}

\section{Introduction}
Recent advances in digital microfluidic (DMF) biochips have enabled realizations of a variety of laboratory assays on a tiny chip for automatic and reliable analysis of biochemical samples. A DMF biochip consists of a patterned 2D array or a customized layout of electrodes, typically a few square centimeters in size \cite{b16,b32}.  The device can manipulate pico- or femtoliter-sized discrete droplets for the purpose of conducting various fluidic operations under electrical actuations. Typical fluidic operations on a droplet include dispensing, transport, mixing, splitting, heating, incubation, and sensing \cite{b2,b19,b20,b21}. DMF biochips offer significant flexibility and programmability over their continuous-flow counterparts while implementing various assays that mandate high-sensitivity, and low requirement of sample and reagent consumption. One such example is sample preparation, which plays a pivotal role in biochemical laboratory protocols, e.g., in polymerase chain reaction (PCR) \cite{JB2008}, and in other applications in biomedical engineering and life sciences \cite{b27,b28,b29,b30,b31}.	An important step in sample preparation is dilution, where the objective is to prepare a fluid with a desired concentration (or dilution) factor. There are two performance metrics in sample preparation: the number of \textit{mix-split} operations to achieve a concentration factor with a specified accuracy, and the overall reactant usage (equivalently, waste production). The first parameter determines the sample preparation time, whereas the latter is related to the cost of stock solution. An efficient sample preparation algorithm should target to minimize either one or  both of them as far as possible.

In sample preparation, producing chemical and biomolecular concentration gradients is of particular interest. Dilution gradients play essential roles in in-vitro analysis of many biochemical phenomena including growth of pathogens and selection of drug concentration. For example, in drug design, it is important to determine the minimum amount of an antibiotic that inhibits the visible growth of bacteria isolate (defined as minimum inhibitory concentration (MIC)). The drug with the least concentration factor (i.e., with highest dilution) that is capable of arresting the growth of bacteria, is considered as MIC. During the past decade, a variety of automated bacterial identification and antimicrobial susceptibility test systems have been developed, which provide results in only few hours rather than days, compared to traditional overnight procedures \cite{b28}. Typical automated susceptibility methods use an exponential dilution gradient (e.g., $1\%, 2\%, 4\%, 8\%, 16\%$) in which concentration factors ($\mathcal{CF}$) of the given sample are in geometric progression \cite{b27}. Linear dilution gradient (e.g., $15\%, 20\%, 25\%, 30\%, 35\%$), in which the concentration factors of the sample appear in arithmetic progression, offers more sensitive tests. Linear gradients are usually prepared by using continuous-flow microfluidic ladder networks \cite{WJW2010}, or by other networks of microchannels \cite{DCJ2001}, \cite{Jang2011}. Since the fluidic microchannels are hardwired, continuous-flow based diluters are designed to cater to only a pre-defined gradient, and thus they suffer from inflexibility and non-programmability. Also, these methods require a  significant amount of costly stock solutions. In contrast, on a DMF biochip platform, a set of random dilution factors can be easily prepared. However, existing algorithms \cite{TsungTCAD12,b22,REMIA} fail to optimize the cost or performance when a certain gradient pattern is required. 

In digital microfluidics, two types of dilution methods are used: serial dilution and interpolated dilution \cite{RSF2003}. A serial dilution consists of a sequence of simple dilution steps to reduce the concentration of a sample. The source of the dilution sample for each step comes from the diluted sample of the previous step. A typical serial procedure generates an \textit{exponential dilution} profile, in which, a unit volume sample/reagent droplet is mixed with a unit-volume buffer $(0\%)$ droplet to obtain two unit-volume droplets of half the concentration. If a $100\%$ sample/reagent is recursively diluted by a buffer solution, then the $\mathcal{CF}$  of the sample/reagent becomes $\frac{1}{2^d}$ after $d$ steps of mixing and balanced splitting. In each \textit{interpolated dilution} step, two unit-volume droplets with $\mathcal{CF}$  $C_1$ and $C_2$ are mixed to obtain two droplets with $\mathcal{CF}$  $\frac{C_1+C_2}{2}$. Both the dilution methods produce concentration factors whose denominators are integral power of two. Thus, in the $(1:1)$ \textit{mix-split} model, the $\mathcal{CF}$  values are approximated  (rounded off) as an $n$-bit binary fraction, i.e., as $\frac{x}{2^n}$, where $x \in \mathbb{Z}^+$, $0 \leq x \leq 2^n$; $n \in \mathbb{Z}^+$, determines the required accuracy (maximum error in $\mathcal{CF}\leq \frac{1}{2^{n+1}}$) of the target concentration factor.

In this work, we present for the first time, an algorithm to generate any arbitrary linear gradient, on-chip, with minimum wastage, while satisfying a required accuracy in concentration factors. Our algorithm utilizes the underlying combinatorial properties of a linear gradient in order to generate the target set. The corresponding layout design of the biochip is also proposed. We prove theoretical results on maximum storage requirement in the layout. Simulation results on different linear gradients show a significant improvement in sample cost over three algorithms \cite{TsungTCAD12,b22,REMIA} that were used earlier for the generation of multiple concentration factors. 

\section{Background and prior art}

In early days, dilutions were obtained by manually measuring and dispensing solutions with a pipette. With the advent of continuous-flow microfluidic (CMF) biochips, dilution gradients were prepared based on diffusive mixing of two or more streams. The degree of diffusion can be regulated by the flow rate or by channel dimensions. Designs of such gradient generators on a CMF biochip were proposed by Walker et al. \cite{b24} and O'Neill et al. \cite{b25}. The flow rates were adjusted by controlling the channel length, which is proportional to fluidic resistance in each channel. Serial dilution CMF biochips for monotonic and arbitrary gradient were also reported by Lee et al. \cite{b7}. Dertinger et al. have shown how a complex CMF network of microchannels designed for diffusive mixing can be used to generate linear, parabolic, or periodic dilution gradients \cite{DCJ2001}. Recently, design of a 2D combinatorial dilution gradient generator has been reported based on a tree-type microchannel structure and an active injection system \cite{Jang2011}.

Although continuous-flow devices are found to be adequate for many  biochemical applications, they are less suitable for tasks requiring a high degree of flexibility or programmable fluid manipulations. These closed-channel systems are inherently difficult to integrate or scale because the parameters that govern fluid flow depend on the properties of the entire system. Thus, permanently etched microstructures suffer from limited reconfigurability and poor fault tolerance.

\begin{figure}[!h]
\centering
\includegraphics[scale=0.6]{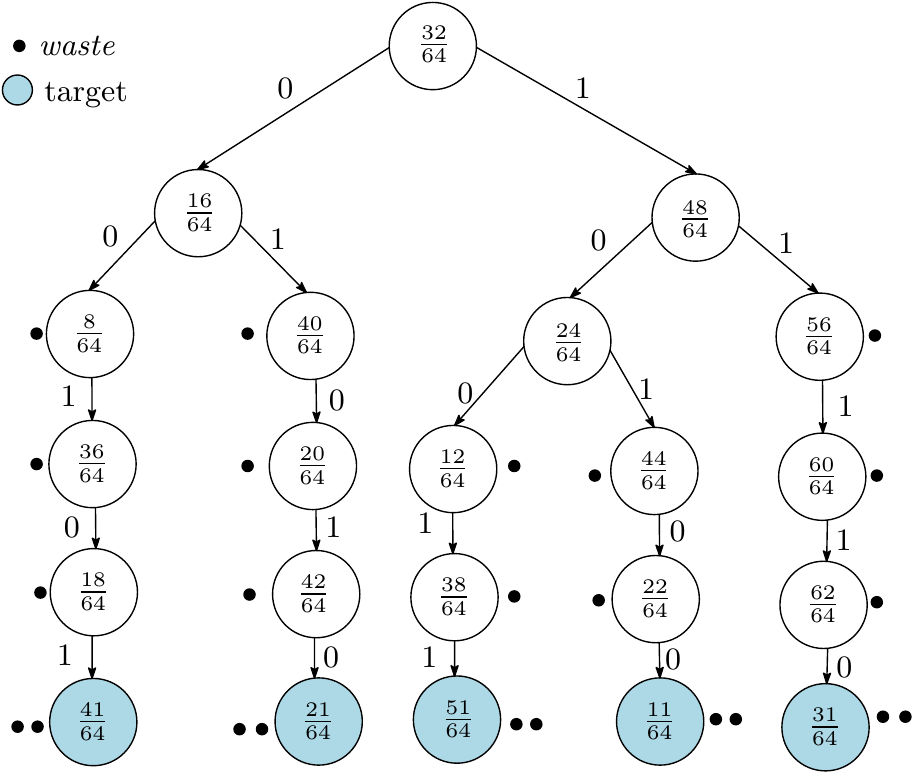}
\caption{ Dilution tree as generated in \cite{TsungTCAD12} }
\label{fig:tsung}
\end{figure}

A DMF biochip typically manipulates discrete fluid droplets on a uniform $2$D array of identical electrodes. Thus the volume of a merged droplet is usually an integral multiple of that of a single droplet (unit volume). It is a challenge to achieve a desired concentration factor $(\mathcal{CF})$ using the fewest number of \textit{mix-split} steps with minimum number of \textit{waste} droplets. A single-target mixing algorithm based on bit-scanning (\textit{BS}) method was proposed by Thies et al. \cite{BS}  considering the $(1:1)$ mixing model. In the special case of diluting a sample, the \textit{BS} method first represents the target $\mathcal{CF}$ as an $n$-bit binary string depending on the required accuracy of $\mathcal{CF}$; it then scans the bits from right-to-left to decide on the sequence of \textit{mix-split} steps, i.e., whether the sample or the buffer droplet is to be mixed with the most recently produced droplet. As an example, any path from the root to a leaf node in Fig. \ref{fig:tsung} represents an execution sequence of the \textit{BS} method. However, it produces one \textit{waste} droplet at each \textit{mix-split} step except the last one. In order to achieve a target $\mathcal{CF}$  with a maximum error of $\frac{1}{2^{n+1}}$, the  dilution process is to be repeated through at most $n$ \textit{mix-split} steps. Thus, depending on the required accuracy level of the target concentration, the value of $n$ is chosen. The \textit{DMRW} method \cite{DMRW} generates a single dilution of a sample using a binary search method that reduces the number of \textit{waste} droplets significantly compared to the \textit{BS} method by reusing the intermediate droplets. Recently, a reagent-saving mixing algorithm for preparing multiple target concentrations was proposed by Hsieh et al. \cite{TsungTCAD12}. For example, the dilution tree for the target set $\mathcal{L} = \{\frac{11}{64}, \frac{21}{64},\frac{31}{64},\frac{41}{64},\frac{51}{64}\}$ is shown in Fig. \ref{fig:tsung}, where the a black dot represents an available droplet (output or waste) \cite{TsungTCAD12}. 

Recently, another method for generating droplets with multiple $\mathcal{CF}$s without using any intermediate storage is reported by Mitra et al. based on de Bruijn graphs \cite{b22}. The \textit{BS} method was generalized for producing multiple $\mathcal{CF}$s with reduced \textit{mix-split} and \textit{waste} droplets \cite{SukantaISED2012}. If multiple target droplets with the same $\mathcal{CF}$ are required in a protocol, a dilution engine can be used \cite{SudipISED2012}. A reactant minimizing multiple dilution sample preparation algorithm was reported by Huang et al. \cite{REMIA}.

\section{Motivation and our contributions}
Gradients play essential roles in studying many biochemical phenomena in-vitro, including the growth of pathogens and efficacy of drugs. Among various types of dilution profiles, linear gradient is most widely used for biochemical analysis. Several sample preparation methods are available that can be used for generating specified gradients. 

Motivated by an example described by Brassard et al. \cite{b32}, we present an algorithm for producing any arbitrary linear dilution gradient with minimum wastage (reagent consumption is minimum). To illustrate the proposed algorithm we assume that the two boundary concentrations (first and last $\mathcal{CF}$s of the target sequence) are available. If droplets with the two boundary $\mathcal{CF}$s are not supplied, we can prepare them by diluting the original $(100\%)$ sample with a buffer $(0\%)$ following an earlier algorithm \cite{BS}, \cite{DMRW}. A simple observation that motivates us to design the proposed linear gradient generator is the following: mixing two non-consecutive $\mathcal{CF}$s, which are separated by an odd number of elements of the gradient sequence, produces the median of the two concentrations. This special property follows from the simple fact that the $\mathcal{CF}$ values in the linear gradient sequence are in arithmetic progression. This property is used to design our algorithm for producing the gradient with no wastage. Moreover, only the concentrations that are elements of the gradient set will be generated during this process. 

\section{Problem formulation}
The problem of linear dilution sample preparation can be formulated as follows. Let $\mathcal{L}=\{\frac{a}{2^n},\frac{a+d}{2^n},\frac{a+2d}{2^n},\ldots,\frac{a+2^kd}{2^n}\}$ be a linear gradient of targets to be generated from $\frac{0}{2^n}$ and $\frac{2^n}{2^n}$, i.e., $|\mathcal{L}|=2^k+1$. Our objective is to generate all $\mathcal{CF}$ values of $\mathcal{L}$ without generating any waste droplets; we assume that a sufficient supply of boundary concentrations ($\frac{a}{2^n}$ and $\frac{a+2^kd}{2^n}$) is available.

\subsection{Zero-waste linear dilution gradient}
The process of generating the target $\mathcal{CF}$s  satisfying a linear dilution gradient can be envisaged as a tree structure called linear dilution tree (LDT), as described below.
   
\begin{algorithm}
\linesnumbered
\scriptsize
\KwIn{A set of $\mathcal{CF}$s  ($\mathcal{L}$)}
\KwOut{Root of the linear dilution tree}

\If {$\mathcal{L}$ contains only one $\mathcal{CF}$}
{
	Create a leaf $v$ storing this point\;
}
\Else
{
	Let $C_{mid}$ be the median of $\mathcal{L}$\;
	Set $L_{left} =$ $\mathcal{CF}$s less than $C_{mid}$ in $\mathcal{L}$\; 
	Set $L_{right} =$ $\mathcal{CF}$s greater than $C_{mid}$ in $\mathcal{L}$\;
	$v_{left} = LDT(L_{left})$\;
	$v_{right} = LDT(L_{right})$\;
	Create a node $v$ storing $C_{mid}$\;
	Make $v_{left}$ the left child of $v$\;
	Make $v_{right}$ the right child of $v$\;
	
}
return $v$\;
\caption{Build linear dilution tree (\textit{LDT})}
\end{algorithm}

A linear dilution tree (\textit{LDT}) is a complete binary search tree having $2^k-1$ nodes, where each node  represents a $\mathcal{CF}$ value in the target set $\mathcal{L}$, where $|\mathcal{L}|=2^k+1$. Thus, the tree will have a depth of $(k-1)$, where the root is assumed to be at depth $0$. 

Algorithm 1 builds \textit{LDT} from the input target set, on which Algorithm 2 described below, will be run to produce the droplets in the target set $\mathcal{L}$.
\begin{algorithm}
\scriptsize
\linesnumbered
\KwIn{$\mathcal{L}=\{\frac{a}{2^n},\frac{a+d}{2^n},\frac{a+2d}{2^n},\ldots,\frac{a+2^kd}{2^n}\}$}
\KwOut{Output ordering of $\mathcal{CF}$s in $\mathcal{L}$}

$\mathcal{L'} = \{\frac{a+d}{2^n},\frac{a+2d}{2^n},\ldots,\frac{a+(2^k-1)d}{2^n}\}$\;
\Comment {$\frac{a}{2^n}$ and $\frac{a+2^kd}{2^n}$ have unlimited supply}\;
$root = LDT(\mathcal{L'})$\;
return $postorder(root)$\;
\Comment {$postorder(root)$ is the post-order traversal of the binary tree}\; 

\caption{Linear dilution}
\label{alg:ldt}
\end{algorithm}

\subsection{An illustrative example}
\begin{figure}[!h]
\centering
\includegraphics[scale=0.5]{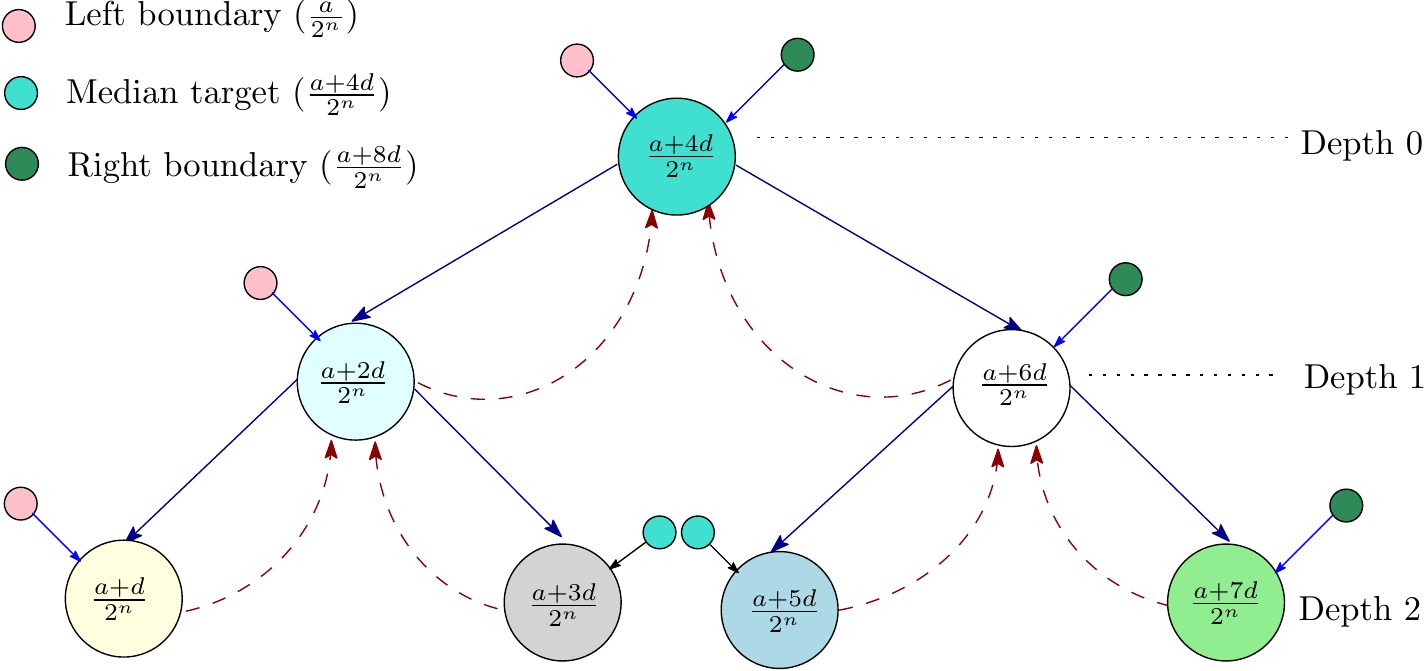}
\caption{Linear dilution tree}
\label{fig:ldt}
\end{figure}

%In this section, we design a digital microfluidic platform where one can easily generate a set of linearly varying concentrations. The proposed method generates only two boundary concentrations from which all other concentrations can be generated.

As an illustration, let us consider $k=3$. Let $\mathcal{T}_{linear}=\{\frac{a}{2^n},\frac{a+d}{2^n},\frac{a+2d}{2^n},\ldots,\frac{a+2^3d}{2^n}\}$ be a linear gradient of targets to be generated from $\frac{0}{2^n}$ and $\frac{2^n}{2^n}$, i.e., $|\mathcal{T}_{linear}|=2^3+1=9$.  The corresponding (\textit{LDT}) is shown in Fig. \ref{fig:ldt}, which is generated by Algorithm 1. We traverse the tree in depth-first order and produce the droplets in a post-order mixing sequence. We assume that the two boundary $\mathcal{CF}$s  $\frac{a}{2^n}$ and $\frac{a+2^3d}{2^n}$ are supplied. Initially, we generate two droplets with $\mathcal{CF}$  $\frac{a+4d}{2^n}$ by mixing one droplet of $\frac{a}{2^n}$ and $\frac{a+8d}{2^n}$ each (represented as the root in Fig. \ref{fig:ldt}). One of these droplets is stored and the other one is mixed with $\frac{a}{2^n}$ to produce two droplets of $\frac{a+2d}{2^n}$. Again, one of them is stored and the other one is mixed with $\frac{a}{2^n}$ to generate two droplets of $\frac{a+d}{2^n}$ (leftmost leaf), out of which one droplet is sent to the output and the other one is stored. Next, the two droplets with $\mathcal{CF}$  $\frac{a+2d}{2^n}$ and $\frac{a+4d}{2^n}$, which were stored in the first two steps, are mixed to produce two droplets of $\frac{a+3d}{2^n}$. One of them is sent to the output; the remaining one is mixed with the one with $\mathcal{CF}$   $\frac{a+d}{2^n}$ stored in the third step. This step regenerates two droplets with $\mathcal{CF}$  $\frac{a+2d}{2^n}$, which were consumed in earlier steps. One of them is stored again and now the other one is transported to the output. Similar \textit{mix-split} sequences are performed on the right half of  \textit{LDT} in post-order fashion, and finally, two droplets of $\frac{a+4d}{2^n}$ (represented as the root) are regenerated by mixing $\frac{a+6d}{2^n}$ with $\frac{a+2d}{2^n}$. It may be observed that {\it no waste droplet} is produced for generating the entire linear dilution sequence $\mathcal{T}_{linear}$. Only one droplet for every non-boundary $\mathcal{CF}$ value in the gradient is produced, excepting the median one, for which two droplets are produced.

The following observations are now immediate. %However, due to page limit, we do not include proofs here.
\begin{observation}
The droplets with boundary $\mathcal{CF}$s   are used only along the leftmost and the rightmost root-to-leaf path in \textit{LDT}.
\label{obs:1}
\end{observation}

\begin{observation}
The droplets with the $\mathcal{CF}$ values corresponding to each internal node of \textit{LDT} are used in subsequent mixing operations after their production and are regenerated later for replenishment.
\label{obs:2}
\end{observation}

\begin{lemma}
The number of copies of each droplet generated at depth $i$ during the process is $2^{k-i}+2$ when $i < k$, and is $2$ when $i = k$ (leaf node), where $|\mathcal{L}|=2^{k+1}+1$.
\label{lem:1}
\end{lemma}
\begin{proof}
We proof the lemma using induction on $k$, i.e., the target set size $|\mathcal{L}|=2^{k+1}+1$.

Basis: For $k=1$, $|\mathcal{L}|=2^{1+1}+1=5$; in this case, we need to generate $3$ $\mathcal{CF}$ values from $2$ boundary concentrations. Fig. \ref{fig:tree_structure} shows the linear dilution tree.  The total number of droplets corresponding to the median target $\mathcal{CF}$ (at the root, depth $=0$) is $4(=2+2)$ and target concentration at depth $1$ is 2 each. Hence the Lemma \ref{lem:1} is true for $k=1$.
\begin{figure}[!h]
\centering
\includegraphics[scale=0.7]{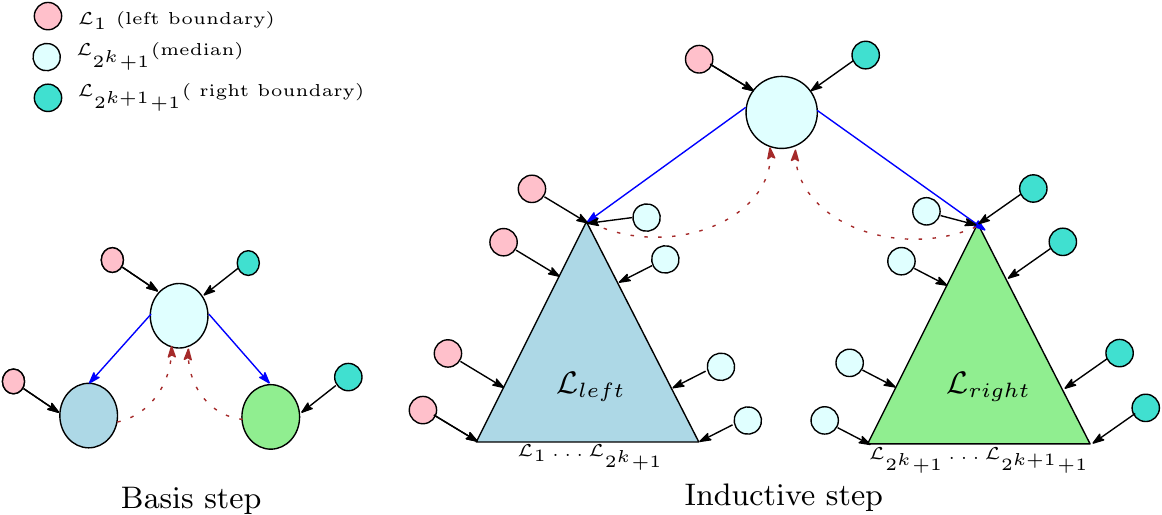}
\caption{}
\label{fig:tree_structure}
\end{figure}

Induction hypothesis: Assume the statement is true for all $m\leq k-1$. 

Inductive steps: Consider the target set $\mathcal{L}$ of size $2^{k+1}+1$ i.e., $m=k$. One can split $\mathcal{L}$ into three parts: $\mathcal{L}_{left}$ that contains the first $2^k+1$ targets of $\mathcal{L}$ i.e., $\mathcal{L}_{left} = \{\mathcal{L}_1,\ldots,\mathcal{L}_{2^k+1}\}$; $\mathcal{L}_{median}=\{\mathcal{L}_{2^k+1}\}$; $\mathcal{L}_{right} = \{\mathcal{L}_{2^k+1},\ldots,\mathcal{L}_{2^{k+1}+1}\}$. The elements in $\mathcal{L}_{left}$ can be generated by using $\mathcal{L}_{1}$ and $\mathcal{L}_{2^k+1}$ as boundary targets. Similarly, those in $\mathcal{L}_{right}$ can be generated by using $\mathcal{L}_{2^k+1}$ and $\mathcal{L}_{2^{k+1}+1}$ as boundary targets. One can easily generate $\mathcal{L}$ by using $\mathcal{L}_1$ and $\mathcal{L}_{2^{k+1}+1}$ as boundary targets (see Fig. \ref{fig:tree_structure}). By induction hypothesis, the number of each droplet generated during the process at depth $i$ of $\mathcal{L}_{left}$ and $\mathcal{L}_{right}$  is $2^{k-1-i}+2$ when $i < k-1 $, and is $2$ when $i = k-1$. Ignoring the regeneration part, the number of each droplet generated during the process at depth $i$ of $\mathcal{L}_{left}$ and $\mathcal{L}_{right}$  is $2^{k-1-i}$ when $i < k-1 $, and is $2$ at depth $k-1$. From Observation \ref{obs:1} it follows that $\mathcal{L}_{median}$ is used only in the rightmost path of $\mathcal{L}_{left}$ and in the leftmost path of $\mathcal{L}_{right}$ as shown in Fig. \ref{fig:tree_structure}. By inductive hypothesis, the total number of droplets generated (ignoring the regeneration part) is $2\times (\sum_{i=0}^{k-2} 2^{k-1-i}+2) = 4\times (\sum_{i=0}^{k-2} 2^{i}+1)=4\times(2^{k-1}-1+1)=4\times 2^{k-1}$. Hence, the required number of $\mathcal{L}_{median}$ droplets is $\frac{4\times 2^{k-1}}{2}=2^{k}$. Since the number of regenerated droplets is 2, the total number of droplets generated at the root (depth $=0$) is $2^{k}+2$. This completes the proof.  
\end{proof}

\begin{lemma}
The number of each boundary droplet required is $2^k$ for $k\geq 1$, where $|\mathcal{L}|=2^{k+1}+1$.
\label{lem:2}
\end{lemma}
\begin{proof}
From Observation \ref{obs:1} it follows that boundary droplets are needed only for the nodes lying on the leftmost and rightmost paths of \textit{LDT}. Note that the regeneration process for an internal node does not require any boundary droplet. So the number of droplets generated excluding regeneration, is $2^{k-i}$ at depth $i$ and is $2$ at depth $k$, along the left- or rightmost path in \textit{LDT}. The total number of droplets along these paths is $\sum_{i=0}^{k-1} 2^{k-i}+2$. Hence, the total number of required boundary droplets will be given by
\begin{equation*}
%\begin{aligned}
\frac{\sum_{i=0}^{k-1} 2^{k-i}+2}{2}
=\left(\sum_{i=0}^{k-1} 2^i\right)+1  
= (2^k-1)+1
= 2^k
%\end{aligned}
\end{equation*}
\end{proof}

\begin{theorem}
Algorithm 2 generates a linear dilution gradient $|\mathcal{L}|(=2^{k+1}+1$) in $2^{k-1}(k+4) - 1$ $(1:1)$ mix-split steps without producing any waste droplets, when $2^k$ droplets of each boundary $\mathcal{CF}$ are supplied.
\label{thm:1}
\end{theorem}
\begin{proof}
The \textit{LDT} has $2^{k+1}-1$ nodes including $2^k$ leaf nodes. Each leaf node requires only one \textit{mix-split} operation. By Lemma \ref{lem:1} the number of each droplet generated at depth $i$ is $2^{k-i}+2$ for $0\leq i\leq k-1$, where the constant 2 accounts for its regeneration from its two children. Regeneration requires $2^k-1$ \textit{mix-split} steps. Hence the total number of \textit{mix-split} operations will be
\begin{equation*}
\begin{aligned}
& 2^k +\frac{\sum_{i=0}^{k-1} 2^{k-i}\times 2^i}{2} + 2^k-1 \\
=& 2^k +\frac{\sum_{i=0}^{k-1} 2^k}{2} + 2^k-1 \\
=& 2^{k+1} +k.2^{k-1} - 1\\
=& 2^{k-1}(k+4) - 1\\
\end{aligned}
\end{equation*}
The fact that no waste droplet is generated in this process follows easily by counting the number of input droplets (Lemma 2) and the output droplets.
\end{proof}
\label{obs:3}
\begin{observation}
The $\mathcal{CF}$ values of the gradient excluding the two boundary $\mathcal{CF}$s appear at the output of the generator in accordance to the post-order traversal sequence of \textit{LDT}.
\end{observation}

The following theorem provides an upper bound on the storage requirement during gradient generation.
\begin{theorem}
Algorithm 2 requires at most $2k$ storage electrodes at any instant of time, where $|\mathcal{L}|=2^{k+1}+1$.
\label{thm:2}
\end{theorem}
\begin{proof}
We proof the lemma using induction on $k$.

Basis: For $k=1$, $|\mathcal{L}|=2^{1+1}+1=5$. One needs to generate three $\mathcal{CF}$s from $2$ boundary concentrations. It is easy to check that we require at most 2 intermediate storage elements in this case. Hence the theorem is true for $k=1$.

Inductive hypothesis: Assume the statement is true for $k-1$.

Inductive steps: Consider the target set $\mathcal{L}$ of size $2^{k+1}+1$. One can split $\mathcal{L}$ into three parts: $\mathcal{L}_{left}$ that contains the first $2^k+1$ targets of $\mathcal{L}$, i.e., $\mathcal{L}_{left} = \{\mathcal{L}_1,\ldots,\mathcal{L}_{2^k+1}\}$; $\mathcal{L}_{median}$ that contains the median target of $\mathcal{L}$; $\mathcal{L}_{right} = \{\mathcal{L}_{2^k+1},\ldots,\mathcal{L}_{2^{k+1}+1}\}$ (see Fig. \ref{fig:tree_structure}). By inductive hypothesis, the left subtree requires $2(k-1)$ storage. Additionally, we need to store one droplet of $\mathcal{CF}$($\mathcal{L}_{2^k+1}$) corresponding to the root. So, a total of $2(k-1) +1=2k-1 $ storage is required in order to generate all the $\mathcal{CF}$s on the left subtree. 

When we generate the target set for the right subtree, we need to store the root $\mathcal{CF}$ of the left subtree for regeneration purpose. By analogous argument, we can claim the right subtree requires $2k-1$ storage. Hence, the total number of storage required is $(2k-1) +1 = 2k$ for a linear dilution tree of size $|\mathcal{L}|=2^{k+1}+1$. 
\end{proof}

\subsection{Generation of boundary droplets}
\begin{figure}[!h]
\centering
\includegraphics[scale=0.6]{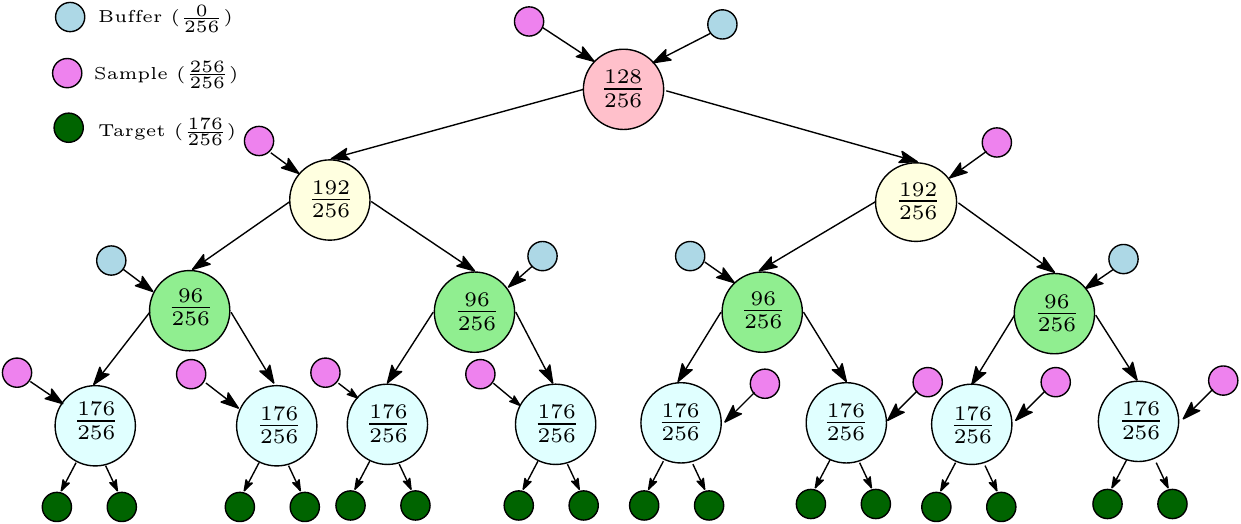}
\caption{Generation of a target $\mathcal{CF}$ with multiple demand}
\label{fig:dil-engine}
\end{figure}

In order to produce a dilution gradient of size $2^{k+1}+1$, we need to supply $2^k$ droplets for each boundary $\mathcal{CF}$ (Lemma \ref{lem:2}). Here, we demonstrate how a low-cost dilution engine  \cite{SudipISED2012} can be integrated on-chip for this purpose. We will illustrate the technique using the following example. Let $\frac{176}{256}$ be a boundary $\mathcal{CF}$. The corresponding dilution tree is shown in the Fig. \ref{fig:dil-engine}. Each root-to-leaf path represents a sequence of \textit{mix-split} operations needed to generate the droplet by applying the \textit{BS} method. One can store the intermediate droplets into a stack and generate two target droplets each time by repeatedly mixing the droplet on top of the stack with either sample or buffer as needed. The dilution tree in Fig. \ref{fig:dil-engine} has $16$ target droplets generated therein. To admit a maximum error of $\frac{1}{2^{n+1}}$, an $n$-depth dilution tree would suffice, and hence the number of storage elements needed to produce multiple droplets with the given $\mathcal{CF}$ will be at most $n-1$ \cite{SudipISED2012}.

\subsection{Generation of linear gradient sequence of arbitrary length}

If the number of elements in the gradient is not of the form $2^{k+1}+1$, the above procedure needs certain modification. Let $Bin(x,m)$ denote the $m$-bit binary representation of $x$ and $ZC(Bin(x,m))$ denote the number of $0$'s between leftmost and rightmost $1$ in it. To illustrate the modification, let us assume that the target set be $\mathcal{L} = \{\mathcal{L}_0,\mathcal{L}_1,\ldots,\mathcal{L}_{10}\}$, i.e., $|\mathcal{L}| = 11$. We consider another target set $\mathcal{L'} =  \{\mathcal{L}_0,\mathcal{L}_1,\ldots,\mathcal{L}_{16}\}$ of size $2^{2+2}+1 = 17 $. Fig. \ref{fig:partial-tree} shows the dilution tree for $\mathcal{L'}$, where the extra part of the tree is not generated (shown as dotted). Clearly, $2^{2+1}+1 < |\mathcal{L}| <2^{2+2}+1$, i.e., $k =2$. Now, $Bin(|\mathcal{L}|-1,4) = 1010$ and the number of \textit{waste droplets} is equal to $ZC(Bin(10,4)) = 1$ ($\mathcal{L}_{12}$).

The following theorem can be proved easily.

\begin{figure}[!h]
\centering
\begin{subfigure}{.25\textwidth}
  \centering
  \includegraphics[scale=.7]{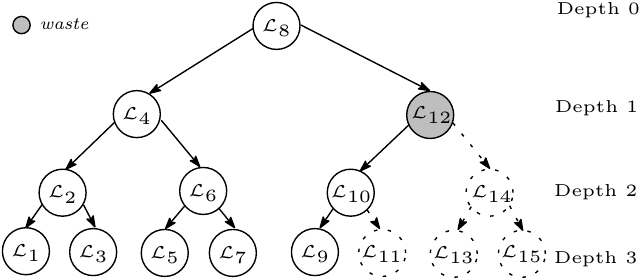}
  \caption{Dilution tree for $\mathcal{L'}$}
  \label{fig:partial-tree}
\end{subfigure}%
\begin{subfigure}{.25\textwidth}
  \centering
  \includegraphics[scale=.4]{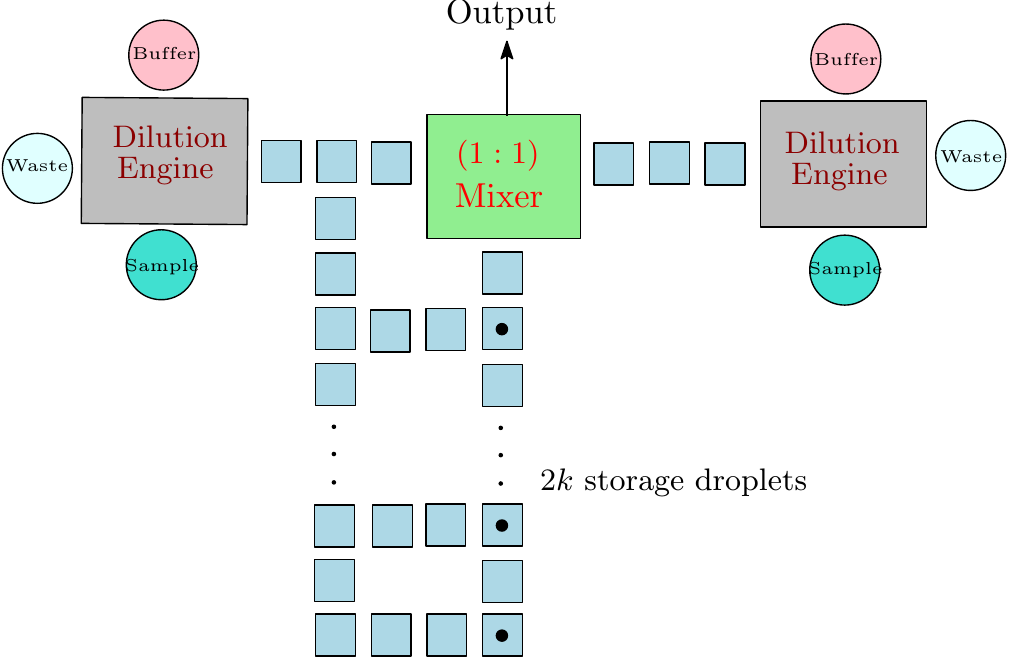}
  \caption{Architectural layout}
  \label{fig:arch}
\end{subfigure}
\caption{ (a)-Dilution tree for $\mathcal{L'}$ (b)-Architectural layout}
%\label{fig:test}
\end{figure}

\begin{theorem}
The number of \textit{waste droplets} produced while generating $\mathcal{L}$ is equal to $ZC(Bin(|\mathcal{L}| - 1,k+2))$, where $2^{k+1}+1 < |\mathcal{L}| < 2^{k+2} +1$.\hfill $\Box$
\label{thm:3}
\end{theorem}

\section{Architecture}

An architectural layout for producing a linear dilution gradient is shown in Fig. \ref{fig:arch}. If boundary $\mathcal{CF}$s  other than $0\%$ and $100\%$ are needed, a dilution engine \cite{SudipISED2012} can be used for generating them. We provide two dilution engines for generating the boundary $\mathcal{CF}$s, which can be run in parallel to reduce sample preparation time. Each dilution engine is equipped with a stack of $(n-1)$ storage droplets in order to increase the throughput of boundary droplets and to reduce the number of \textit{waste} droplets. The detailed layout of the dilution engine can be found elsewhere \cite{SudipISED2012}. To generate the gradient part, we use one $(1:1)$ mix-split module and $2k$ additional storage cells (Theorem \ref{thm:2}). Thus, to produce a gradient of size $|\mathcal{L}|=2^{k+1}+1$, with a maximum error of $\frac{1}{2^{n+1}}$ in each of the target $\mathcal{CF}$, one needs a total of $2(n+k-1)$ storage electrodes. The overall execution time for generating the gradient can further be minimized by adopting a scheduling algorithm \cite{OGB2012} for the best utilization of resources. 
 
\section{Simulation results}

We have performed extensive simulation on  various target sets (Table \ref{tab:TS}) and calculated the number of $(1:1)$-\textit{mix-split} steps and \textit{waste} droplets. We have compared our results with the methods of Mitra et al. \cite{b22}, Hsieh et al. \cite{TsungTCAD12}, and Huang et al. \cite{REMIA}. The results are shown in Table \ref{tab:mix-split-and-waste}.  The number of \textit{waste} droplets in \textit{LDT} for the proposed method is shown within  parentheses in Table \ref{tab:mix-split-and-waste} along with total \textit{mix-split} steps ($\mathcal{M}$) and \textit{waste} ($\mathcal{W}$) droplets. 

\begin{table}[h]
\caption{Target set ($n=10$)} 
\centering 
{\fontsize{8pt}{6pt}\selectfont 
\begin{tabular}{|c|c|c|}
\hline
& TS size& Target set\\ \hline 
$TS_1$ & 5&$\{\frac{50}{2^{10}},\frac{70}{2^{10}},\ldots,\frac{130}{2^{10}}\}$  \\ \hline
$TS_2$ & 10&$\{\frac{110}{2^{10}},\frac{120}{2^{10}},\ldots,\frac{200}{2^{10}}\}$ \\ \hline
$TS_3$ & 17&$\{\frac{20}{2^{10}},\frac{70}{2^{10}},\ldots,\frac{820}{2^{10}}\}$ \\ \hline 
$TS_4$ & 20&$\{\frac{40}{2^{10}},\frac{70}{2^{10}},\ldots,\frac{610}{2^{10}}\}$ \\ \hline
\end{tabular}
}
\label{tab:TS}
\end{table}

\begin{table}[!h]
\caption{The number of \textit{mix-split} steps \& \textit{waste} droplets}
\centering 
{\fontsize{8pt}{6pt}\selectfont \centering
\begin{tabular}{|c|c|c|c|c|c|c|c|c|}
\hline
$TS$ & \multicolumn{2}{|c|}{Proposed} & \multicolumn{2}{|c|}{Mitra et al.} & \multicolumn{2}{|c|}{Hsieh et al.} & \multicolumn{2}{|c|}{Huang et al.}\\
\hline 
$n=10$&$\mathcal{M}$&$\mathcal{W}$&$\mathcal{M}$&$\mathcal{W}$&$\mathcal{M}$&$\mathcal{W}$&$\mathcal{M}$&$\mathcal{W}$\\ \hline \hline
$TS_1$&24&14 (0)&33&28&37&28&32&12\\
$TS_2$&40&13 (2)&64&54&55&36&64&22\\
$TS_3$&57&12 (0)&88&71&82&50&104&39\\
$TS_4$&65&10 (2)&108&88&95&58&126&49\\
\hline
\end{tabular}
}
\label{tab:mix-split-and-waste}
\end{table}

\begin{figure}[!h]
\centering
\begin{subfigure}{.25\textwidth}
  \centering
  \includegraphics[scale=.37]{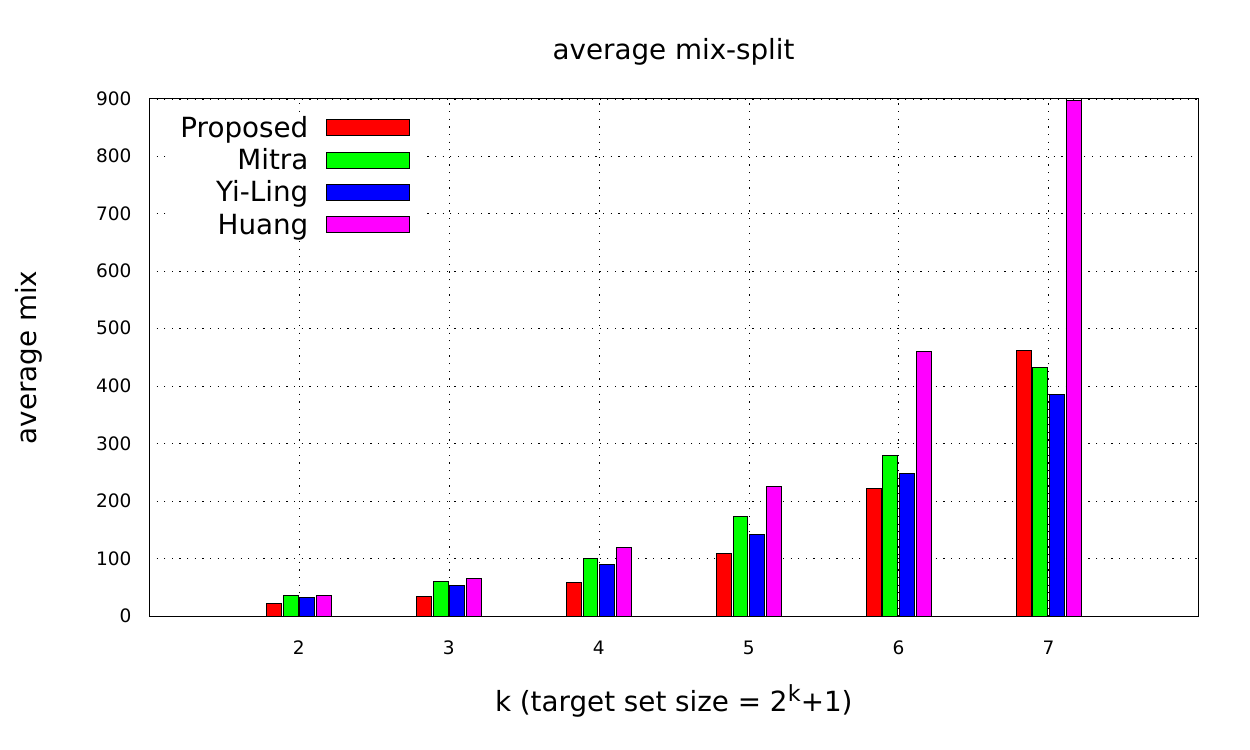}
  \caption{\textit{mix-split}}
  \label{fig:sub1}
\end{subfigure}%
\begin{subfigure}{.25\textwidth}
  \centering
  \includegraphics[scale=.37]{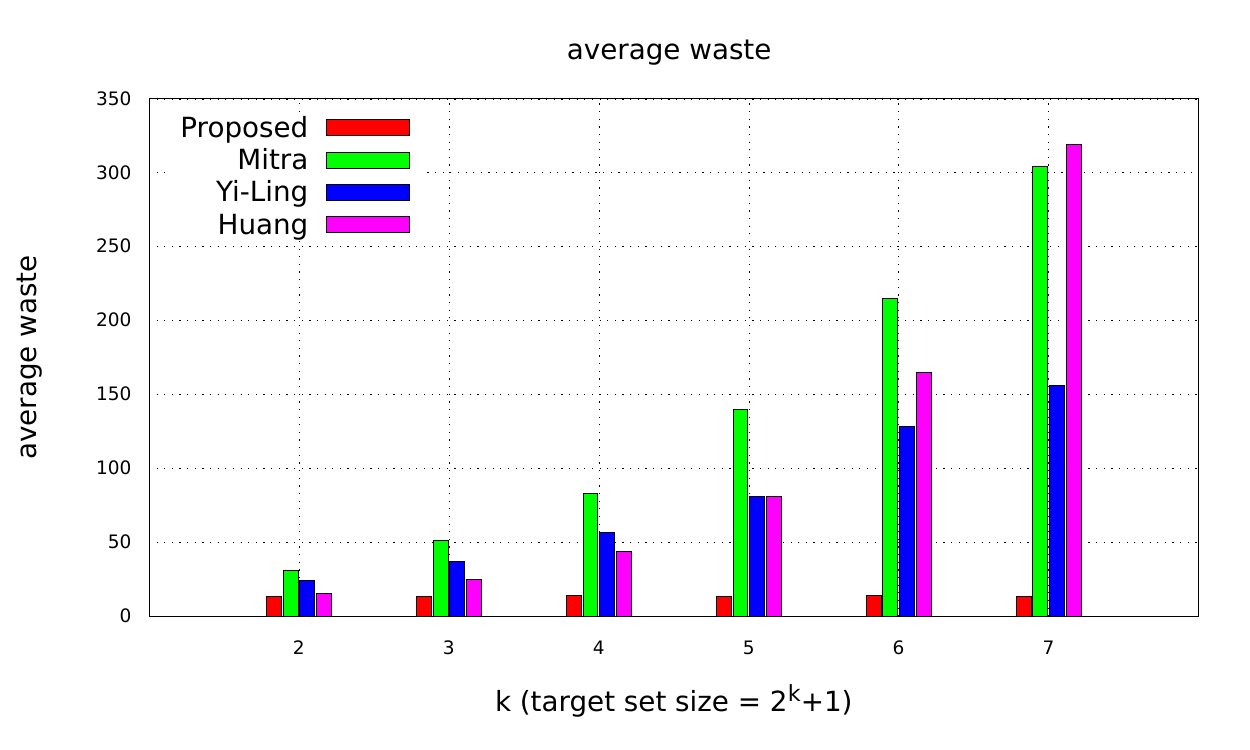}
  \caption{\textit{waste}}
  \label{fig:sub2}
\end{subfigure}
\caption{Comparison of the proposed method with Mitra et al. \cite{b22}, Hsieh et al. \cite{TsungTCAD12}, and Huang et al. \cite{REMIA}}
\label{fig:test}
\end{figure}

In our experiments, we have considered 6 different linear dilution gradient sets $\mathcal{L}$ of size $2^k+1$ for $k=2,\ldots,7$. For each $k$, we have chosen $100$ random sets in the range of $\frac{1}{1024}$ and $\frac{1023}{1024}$, assuming $n=10$. Hence the error in concentration factors will be at most $\frac{1}{2048}$. Assuming $100\%$ sample and $0\%$ buffer as the two boundary $\mathcal{CF}$s, we have counted the total number of $(1:1)$-\textit{mix-split} steps and \textit{waste} droplets considering both the dilution engines and the gradient generator.  Comparative results with respect to earlier methods are shown as histograms in Fig. \ref{fig:test}, where the horizontal axis indicates the size of the target set ($|\mathcal{L}| = 2^k+1$) for $k=2,\ldots,7$, and the vertical axis represents the average number of \textit{mix-split} steps and \textit{waste} droplets required in these methods. Note that the most of the \textit{waste} droplets that are generated in our method correspond to those produced by the dilution engines. We observe that our method produces a significantly fewer number of \textit{waste} droplets compared to all the three earlier methods. Further, the proposed method performs better in terms of the number of \textit{mix-split} steps up to a target set of size $65$, i.e., up to $k=6$.

\section{Conclusions}

We have presented an algorithm for generating linear dilution gradients on a digital microfluidic platform. When the boundary concentration factors of a gradient of size $|\mathcal{L}| = 2^k+1$ are supplied, our method produces the rest without generating any waste droplet, thereby saving costly stock solutions. For other gradient sizes, it produces only a few waste droplets. We have also designed a suitable layout architecture to implement the generator on-chip. Our method is adaptive to the size of dilution gradient as well to the desired accuracy of concentration factor. Thus, the proposed approach will provide a flexible and programmable environment for catering to any need of arbitrary linear gradient during sample preparation. Generation of other dilution gradients such as parabolic or sinusoidal with a digital microfluidic biochip may be studied as a future problem.

\begin{spacing}{0.8}
\bibliographystyle{IEEEtran}

\begin{thebibliography}{10}
\providecommand{\url}[1]{#1}
\csname url@samestyle\endcsname
\providecommand{\newblock}{\relax}
\providecommand{\bibinfo}[2]{#2}
\providecommand{\BIBentrySTDinterwordspacing}{\spaceskip=0pt\relax}
\providecommand{\BIBentryALTinterwordstretchfactor}{4}
\providecommand{\BIBentryALTinterwordspacing}{\spaceskip=\fontdimen2\font plus
\BIBentryALTinterwordstretchfactor\fontdimen3\font minus
  \fontdimen4\font\relax}
\providecommand{\BIBforeignlanguage}[2]{{%
\expandafter\ifx\csname l@#1\endcsname\relax
\typeout{** WARNING: IEEEtran.bst: No hyphenation pattern has been}%
\typeout{** loaded for the language `#1'. Using the pattern for}%
\typeout{** the default language instead.}%
\else
\language=\csname l@#1\endcsname
\fi
#2}}
\providecommand{\BIBdecl}{\relax}
\BIBdecl

\bibitem{b16}
M.~G. Pollack, R.~B. Fair, and A.~D. Shenderov, ``Electrowetting-based
  actuation of liquid droplets for microfluidic applications,'' \emph{Applied
  Physics Letters}, vol.~77, no.~11, pp. 1725 --1726, Sept. 2000.

\bibitem{b32}
D.~Brassard, L.~Malic, C.~Miville-Godin, F.~Normandin, and T.~Veres, ``Advanced
  {EWOD}-based digital microfluidic system for multiplexed analysis of
  biomolecular interactions,'' in \emph{Micro Electro Mechanical Systems
  (MEMS)}, Jan. 2011, pp. 153--156.

\bibitem{b2}
R.~B. Fair, A.~Khlystov, T.~D. Tailor, V.~Ivanov, R.~D. Evans, V.~Srinivasan,
  V.~K. Pamula, M.~G. Pollack, P.~B. Griffin, and J.~Zhou, ``Chemical and
  biological applications of digital-microfluidic devices,'' \emph{IEEE Design
  {\&} Test of Computers}, vol.~24, no.~1, pp. 10--24, 2007.

\bibitem{b19}
K.~Chakrabarty and F.~Su, \emph{Digital Microfluidic Biochips - Synthesis,
  Testing, and Reconfiguration Techniques}.\hskip 1em plus 0.5em minus
  0.4em\relax CRC Press, 2007.

\bibitem{b20}
S.~K. Cho, H.~Moon, and C.-J. Kim, ``Creating, transporting, cutting, and
  merging liquid droplets by electrowetting-based actuation for digital
  microfluidic circuits,'' \emph{Microelectromechanical Systems, Journal of},
  vol.~12, no.~1, pp. 70--80, Feb. 2003.

\bibitem{b21}
Y.~Fouillet, D.~Jary, C.~Chabrol, P.~Claustre, and C.~Peponnet, ``Digital
  microfluidic design and optimization of classic and new fluidic functions for
  lab on a chip systems,'' \emph{Microfluidics and Nanofluidics}, vol.~4, pp.
  159--165, 2008.

\bibitem{JB2008}
J.~Berthier, \emph{Micro-Drops and Digital Microfluidics}.\hskip 1em plus 0.5em
  minus 0.4em\relax Elsevier, Feb. 2008.

\bibitem{b27}
S.~Sugiura, K.~Hattori, and T.~Kanamori, ``Microfluidic serial dilution
  cell-based assay for analyzing drug dose response over a wide concentration
  range,'' \emph{Analytical Chemistry}, vol.~82, no.~19, pp. 8278--8282, 2010.

\bibitem{b28}
G.~V. Doern, R.~Vautour, M.~Gaudet, and B.~Levy, ``Clinical impact of rapid in
  vitro susceptibility testing and bacterial identification,'' \emph{Journal of
  Clinical Microbiology}, vol.~32, no.~7, pp. 1757--1762, July 1994.

\bibitem{b29}
M.~J. Jebrail and A.~R. Wheeler, ``Digital microfluidic method for protein
  extraction by precipitation,'' \emph{Analytical Chemistry}, vol.~81, no.~1,
  pp. 330--335, 2009.

\bibitem{b30}
X.~Chen, D.~Cui, C.~Liu, H.~Li, and J.~Chen, ``Continuous flow microfluidic
  device for cell separation, cell lysis and dna purification.''
  \emph{Analytica Chimica Acta}, vol. 584, no.~2, pp. 237--243, 2007.

\bibitem{b31}
K.~Lee, C.~Kim, G.~Jung, T.~Kim, J.~Kang, and K.~Oh, ``Microfluidic
  network-based combinatorial dilution device for high throughput screening and
  optimization,'' \emph{Microfluidics and Nanofluidics}, vol.~8, pp. 677--685,
  2010.

\bibitem{WJW2010}
S.~Wang, N.~Ji, W.~Wang, and Z.~Li, ``Effects of non-ideal fabrication on the
  dilution performance of serially functioned microfluidic concentration
  gradient generator,'' in \emph{Nano/Micro Engineered and Molecular Systems
  (NEMS)}, 2010, pp. 169--172.

\bibitem{DCJ2001}
S.~K.~W. Dertinger, D.~T. Chiu, N.~L. Jeon, and G.~M. Whitesides, ``Generation
  of gradients having complex shapes using microfluidic networks,''
  \emph{Analytical Chemistry}, vol.~73, no.~6, pp. 1240--1246, 2001.

\bibitem{Jang2011}
Y.-H. Jang, M.~J. Hancock, S.~B. Kim, S.~Selimovic, W.~Y. Sim, H.~Bae, and
  A.~Khademhosseini, ``An integrated microfluidic device for two-dimensional
  combinatorial dilution,'' \emph{Lab Chip}, vol.~11, pp. 3277--3286, 2011.

\bibitem{TsungTCAD12}
Y.-L. Hsieh, T.-Y. Ho, and K.~Chakrabarty, ``A reagent-saving mixing algorithm
  for preparing multiple-target biochemical samples using digital
  microfluidics,'' \emph{IEEE Trans. on CAD of Integrated Circuits and
  Systems}, vol.~31, no.~11, pp. 1656--1669, 2012.

\bibitem{b22}
D.~Mitra, S.~Roy, K.~Chakrabarty, and B.~B. Bhattacharya, ``On-chip sample
  preparation with multiple dilutions using digital microfluidics,'' in
  \emph{ISVLSI}, 2012, pp. 314--319.

\bibitem{REMIA}
J.-D. Huang, C.-H. Liu, and T.-W. Chiang, ``Reactant minimization during sample
  preparation on digital microfluidic biochips using skewed mixing trees,'' in
  \emph{Computer-Aided Design (ICCAD), 2012 IEEE/ACM International Conference},
  Nov. 2012, pp. 377--383.

\bibitem{RSF2003}
H.~Ren, V.~Srinivasan, and R.~Fair, ``Design and testing of an interpolating
  mixing architecture for electrowetting-based droplet-on-chip chemical
  dilution,'' in \emph{TRANSDUCERS, Solid-State Sensors, Actuators and
  Microsystems}, vol.~1, 2003, pp. 619--622.

\bibitem{b24}
G.~M. Walker, N.~Monteiro-Riviere, J.~Rouse, and O.~T. Neill, ``A linear
  dilution microfluidic device for cytotoxicity assays,'' in \emph{Lab Chip},
  vol.~7, 2010, pp. 226--232.

\bibitem{b25}
A.~O'Neill, N.~Monteiro-Riviere, and G.~Walker, ``A serial dilution
  microfluidic device for cytotoxicity assays,'' in \emph{Engineering in
  Medicine and Biology Society}, Sept. 2006, pp. 2836--2839.

\bibitem{b7}
K.~Lee, C.~Kim, B.~Ahn, R.~Panchapakesan, A.~R. Full, L.~Nordee, J.~Y. Kang,
  and K.~W. Oh, ``Generalized serial dilution module for monotonic and
  arbitrary microfluidic gradient generators,'' \emph{Lab Chip}, vol.~9, pp.
  709--717, 2009.

\bibitem{BS}
W.~Thies, J.~P. Urbanski, T.~Thorsen, and S.~P. Amarasinghe, ``Abstraction
  layers for scalable microfluidic biocomputing,'' \emph{Natural Computing},
  vol.~7, no.~2, pp. 255--275, 2008.

\bibitem{DMRW}
S.~Roy, B.~B. Bhattacharya, and K.~Chakrabarty, ``Optimization of dilution and
  mixing of biochemical samples using digital microfluidic biochips,''
  \emph{IEEE Trans. on CAD of Integrated Circuits and Systems}, vol.~29,
  no.~11, pp. 1696--1708, 2010.

\bibitem{SukantaISED2012}
S.~Bhattacharjee, A.~Banerjee, and B.~B. Bhattacharya, ``Multiple dilution
  sample preparation using digital microfluidic biochips,'' in \emph{ISED},
  2012.

\bibitem{SudipISED2012}
S.~Roy, B.~B. Bhattacharya, S.~Ghoshal, and K.~Chakrabarty, ``Low-cost dilution
  engine for sample preparation in digital microfluidic biochips,'' in
  \emph{ISED}, 2012.

\bibitem{OGB2012}
K.~O'Neal, D.~Grissom, and P.~Brisk, ``Force-directed list scheduling for
  digital microfluidic biochips,'' in \emph{VLSI and System-on-Chip
  (VLSI-SoC)}, 2012.

\end{thebibliography}

\end{spacing}

\end{document}